\documentclass[10pt, conference, compsocconf]{IEEEtran}

\usepackage{listings}

\usepackage{amsmath}
\usepackage{varioref}
\usepackage{amssymb}
\usepackage[algoruled,linesnumbered]{algorithm2e}
\usepackage{textcomp}
\usepackage{graphicx}
\usepackage{palatino}
\usepackage{mathrsfs}
\usepackage{dsfont}
\usepackage{multicol}

\usepackage{caption}
\usepackage{subcaption}

\usepackage{amsthm}

\newtheorem{proposition}{Proposition}
\theoremstyle{definition}
\newtheorem{definition}{Definition}
\newtheorem{example}{Example}

\usepackage{tikz}
\usetikzlibrary{arrows,automata}
\usetikzlibrary{decorations.pathmorphing}
\tikzset{initial text={}}

\SetKwBlock{SubAlgoBlock}{}{end}
\newcommand{\SubAlgo}[2]{#1 \SubAlgoBlock{#2}}

\DeclareMathOperator{\lin}{lin}

\begin{document}

\sloppy

\title{Update Consistency for Wait-free Concurrent Objects}

\author{
\IEEEauthorblockN{Matthieu Perrin, Achour Mostefaoui, and Claude Jard}
  \IEEEauthorblockA{LINA -- University of Nantes,
    Nantes, France\\
    Email: [firstname.lastname]@univ-nantes.fr}
}

\maketitle

\begin{abstract}
In large scale systems such as the Internet, replicating data is an essential feature in order to provide availability and fault-tolerance. Attiya and Welch proved that using strong consistency criteria such as atomicity is costly as each operation may need an execution time linear with the latency of the communication network. Weaker consistency criteria like causal consistency and PRAM consistency do not ensure convergence. The different replicas are not guaranteed to converge towards a unique state. Eventual consistency guarantees that all replicas eventually converge when the participants stop updating. However, it fails to fully specify the semantics of the operations on shared objects and requires additional non-intuitive and error-prone distributed specification techniques.

\vspace{-1mm}
This paper introduces and formalizes a new consistency criterion, called \emph{update consistency}, that requires the state of a replicated object to be consistent with a linearization of all the updates. In other words, whereas atomicity imposes a linearization of all of the operations, this criterion imposes this only on updates. Consequently some read operations may return out-dated values. Update consistency is stronger than eventual consistency, so we can replace eventually consistent objects with update consistent ones in any program. Finally, we prove that update consistency is universal, in the sense that any object can be implemented under this criterion in a distributed system where any number of nodes may crash. 
\end{abstract}


\vspace{-1mm}
\section{Introduction}

\vspace{-1mm}
Reliability of large scale systems is a big challenge when building massive distributed applications
over the Internet. At this scale, data replication is essential to ensure availability and fault-tolerance. In a perfect world, distributed objects should behave as if there is a unique physical shared object that evolves following the atomic operations issued by the participants\footnote{We use indifferently participant or process to designate the computing entities that invoke the distributed object.}. This is the aim of strong consistency criteria such as linearizability and sequential consistency. These criteria serialize all the operations so that they look as if they happened sequentially, but they are costly to implement in message-passing systems. If one considers a distributed implementation of a shared register, the worst-case response time must be proportional to the latency of the network either for the reads or for the writes to be sequentially consistent \cite{lipton1988pram} and for all the operations for linearizability \cite{attiya1994sequential}. This generalizes to many objects \cite{attiya1994sequential}. Moreover, the availability of the shared object cannot be ensured in asynchronous systems where more than a minority of the processes of a system may crash \cite{AttiyaBD95}. In large modern distributed systems such as Amazon's cloud, partitions do occur between data centers, as well as inside data centers \cite{vogels2008eventually}. Moreover, it is economically unacceptable to sacrifice availability. The only solution is then to provide weaker consistency criteria. Several weak consistency criteria have been considered for modeling shared memory such as PRAM \cite{lipton1988pram} or causality \cite{ahamad1995causal}. They expect the local histories observed by each process to be plausible, regardless of the other processes. However, these criteria do not impose that the data eventually converges to a consistent state. Eventual consistency \cite{vogels2008eventually} is another weak consistency criterion which requires that when all the processes stop updating then all replicas eventually converge to the same state.

\vspace{-1mm}
This paper follows the long quest of the (a) strongest consistency criterion (there may exist several incomparable criteria) implementable for different types of objects in an asynchronous system where all but one process may crash (wait-free systems \cite{Herlihy91Wait}). A contribution of this paper consists in proving that weak consistency criteria such as eventual consistency and causal consistency cannot be combined is such systems. This paper chooses to explore the enforcement of eventual consistency. The relevance of eventual consistency has been illustrated many times. It is used in practice in many large scale applications such as Amazon's Dynamo highly available key-value store \cite {decandia2007dynamo}. It has been widely studied and many algorithms have been proposed to implement eventually consistent shared object. Conflict-free replicated data types (CRDT) \cite {shapiro2011conflict} give sufficient conditions on the specification of objects so that they can be implemented. More specifically, if all the updates made on the object commute or if the reachable states of the object form a semi-lattice then the object has an eventually consistent implementation \cite{shapiro2011conflict}. Unfortunately, many useful objects are not CRDTs.

\vspace{-1mm}
The limitations of eventual consistency led to the study of stronger criteria such as strong eventual consistency \cite{shapiro2011comprehensive}. Indeed, eventual consistency requires the convergence towards a \emph{common state} without specifying which states are legal. In order to prove the correctness of a program, it is necessary to fully specify which behaviors are accepted for an object. The meaning of an operation often depends on the context in which it is executed. The notion of \emph{intention} is widely used to specify collaborative editing \cite {sun1998achieving,li2000new}. The intention of an operation not only depends on the operation and the state on which it is done, but also on the intentions of the concurrent operations. In another solution \cite{bieniusa2012optimized}, it is claimed that, it is sufficient to specify what the concurrent execution of all pairs of non-commutative operations should give (e.g. an error state). This result, acceptable for the shared set, cannot be extended to other more complicated objects. In this case, any partial order of updates can lead to a different result. This approach was formalized in \cite {burckhardt2014replicated}, where the concurrent specification of an object is defined as a function of partially ordered sets of updates to a consistent state leading to specifications as complicated as the  implementations themselves. Moreover, a concurrent specification of an object uses the notion of \emph{concurrent events}. In message-passing systems, two events are concurrent if they are produced by different processes and each process produced its event before it received the notification message from the other process. In other words, the notion of concurrency depends on the implementation of an object not on its specification. Consequently, the final user may not know if two events are concurrent without explicitly tracking the underlying messages. A specification should be independent of the system on which it is implemented.

\vspace{-1mm}
\paragraph{Contributions of the paper} for not restricting this work to a given data structure, this paper first defines a class of data types called UQ-ADT for \emph{update-query abstract data type}. This class encompasses all data structures where an operation either modifies the state of the object (update) or returns a function on the current state of the object (query). This class excludes data types such as a stack where the pop operation removes the top of the stack and returns it (update and query at the same time). 
However, such operations can always be separated into a query and an update (lookup\_top and delete\_top in the case of the stack) which is not a problem as, in weak consistency models, it is impossible to ensure atomicity anyway. 
This paper has three main contributions.
\vspace{-2mm}
\begin{itemize}
  \item It proves that in a wait-free asynchronous system, it is not possible to implement eventual and causal consistency for all UQ-ADTs.
  \item It introduces \emph{update consistency}, a new consistency criterion stronger than eventual consistency and for which the converging state must be consistent with a linearization of the updates.
  \item Finally, it proves that for any UQ-ADT object with a sequential specification there exists an update consistent implementation by providing a generic construction. 
\end{itemize}

\vspace{-2mm}
The remainder of this paper is organized as follows. Section \ref{section:formalization} formalizes the notion of consistency criteria and the type of objects we target in this paper. 
Section \ref{section:eventual consistency} recalls the definition of (strong) eventual consistency.
Section \ref {section:pipelined convergence} proves that eventual consistency cannot be combined with causal consistency in wait-free systems. Section \ref {section:update consistency} introduces (strong) update consistency and compares it with (strong) eventual consistency. Section \ref{section:set} compares, through the example of the set, the expressiveness of strong update consistency and strong eventual consistency. Section \ref{section:implementation} presents a generic construction for any UQ-ADT object with a sequential specification. Finally, Section \ref{section:conclusion} concludes the paper.

\vspace{-2mm}
\section{Abstract Data Types and Consistency Criteria}
\vspace{-3mm}
\label{section:formalization}

Before introducing the new consistency criterion, this section formalizes the notion of object and how a consistency criterion is defined. In distributed systems, sharing objects is a way to abstract message-passing communication between processes. The abstract type of these objects has a sequential specification, defined in this paper by a transition system that characterizes the sequential histories allowed for this object. However, shared objects are implemented in a distributed system using replication and the events of the distributed history generated by the execution of a distributed program is a partial order \cite{lamport1978time}. The consistency criterion makes the link between the sequential specification of an object and a distributed execution that invokes it. This is done by characterizing the partially ordered histories of the distributed program that are acceptable. The formalization used in this paper is explained with more details in \cite{RR}.

\vspace{-1mm}
An abstract data type is specified using a transition system very close to Mealy machines \cite{mealy1955method} except that infinite transition systems are allowed as many objects have an unbounded specification. As stated in the Introduction, this paper focuses on "update-query" objects. On the one hand, the updates have a side-effect that usually affects the state of the object (hence all processes), but return no value. They correspond to transitions between abstract states in the transition system. On the other hand, the queries are read-only operations. They produce an output that depends on the state of the object. Consequently, the input alphabet of the transition system is separated into two classes of operations (updates and queries).

\begin{definition}[Update-query abstract data type]
  An update-query abstract data type (UQ-ADT) is a tuple $O = (U, Q_i, Q_o, S, s_0, T, G)$ such that:
  \vspace{-1mm}
  \begin{itemize}
  \item $U$ is a countable set of \emph{update} operations;
  \item $Q_i$ and $Q_o$ are countable sets called \emph{input} and \emph{output} alphabets; 
    $Q = Q_i \times Q_o$ is the set of \emph{query} operations. A query operation $(q_i, q_o)\in Q$ is denoted $q_i/q_o$ (query $q_i$ returns value $q_o$).
  \item $S$ is a countable set of \emph{states};
  \item $s_0 \in S$ is the \emph{initial state};
  \item $T : S\times U\rightarrow S$ is the \emph{transition function};
  \item $G : S\times Q_i\rightarrow Q_o$ is the \emph{output function}.
\end{itemize}
  \vspace{-1mm}
A sequential history is a sequence of operations. An infinite sequence of operations $(w_i)_{i\in \mathbb{N}} \in (U\cup Q)^\omega$ is recognized by $O$ if there exists an infinite sequence of states $(s_i)_{i\ge 1} \in S^\omega$ (note that $s_0$ is the initial state) such that for all $i\in \mathbb{N}$, $T(s_i, w_i) = s_{i+1}$ if $w_i \in U$ or $s_i = s_{i+1}$ and $G(s_i, q_i) = q_o$ if $w_i = q_i/q_o \in Q$. The set of all infinite sequences recognized by $O$ and their finite prefixes is denoted by $L(O)$. Said differently, $L(O)$ is the set of all the sequential histories allowed for $O$.
\end{definition}

\vspace{-3mm}
Along the paper, replicated sets are used as the key example. Three kinds of operations are possible: two update operation by element, namely insertion (I) and deletion (D) and a query operation read (R) that returns the values that belong to the set. Let $\mathit{Val}$ be the support of the replicated set (it contains the values that can be inserted/deleted). At the beginning, the set is empty and when an element is inserted, it becomes present until it is deleted. More formally, it corresponds to the UQ-ADT given in Example \ref{def:set}. 

\vspace{-1mm}
\begin{example}[Specification of the set]\label{def:set}
  Let $\mathit{Val}$ be a countable set, called support. The set object
  $\mathcal{S}_{\mathit{Val}}$ is the UQ-ADT $(U, Q_i, Q_o, S, \emptyset, T, G)$ with:
\vspace{-1mm}
  \begin{itemize}
  \item $ U = \{\mathrm{I}(v), \mathrm{D}(v) : v\in \mathit{Val}\}$;
  \item $ Q_i = \{\mathrm{R}\}$, and $Q_o = S = \mathcal{P}_{<\infty}(\mathit{Val})$ contain all the finite subsets of $\mathit{Val}$;
  \item for all $s\in S$ and $v\in \mathit{Val}$, $G(s, \mathrm{R}) = s$, \\$T(s, \mathrm{I}(v)) = s \cup \{v\}$ and $T(s, \mathrm{D}(v)) = s \setminus \{v\}$.
  \end{itemize}
\vspace{-1mm}
\end{example}

\vspace{-1mm}
The set $U$ of updates is the set of all insertions and deletions of any value of $\mathit{Val}$. The set of queries $Q_i$ contains a unique operation $R$, a read operation with no parameter. A read operation may return any value in $Q_o$, the set of all finite subsets of $\mathit{Val}$. The set $S$ of the possible states is the same as the set of possible returned values $Q_o$ as the read query returns the content of the set object. $I(v)$ (resp. $D(v)$) with $v\in \mathit{Val}$ denotes an insertion (resp. a deletion) operation of the value $v$ into the set object. $R/s$ denotes a read operation that returns the set $s$ representing the content of the set.

\vspace{-1mm}
During an execution, the participants invoke an object instance of an abstract data type using the associated operations (queries and updates). This execution produces a set of partially ordered events labelled by the operations of the abstract data type. This representation of a distributed history is generic enough to model a large number of distributed systems. For example, in the case of communicating sequential processes, an event $a$ precedes an event $b$ in the \emph{program order} if they are executed by the same process in that sequential order. It is also possible to model more complex modern systems in which new threads are created and destroyed dynamically, or peer-to-peer systems where peers may join and leave.

\vspace{-1mm}
\begin{definition}[Distributed History]
  A distributed history is a tuple $H = (U, Q, E, \Lambda, \mapsto)$:
\vspace{-1mm}
  \begin{itemize}
  \item $U$ and $Q$ are disjoint countable sets of \emph{update} and \emph{query} operations, and all queries $q\in Q$ are in the form $q = q_i/q_o$;
  \item $E$ is a countable set of \emph{events};
  \item $\Lambda : E \rightarrow U\cup Q$ is a \emph{labelling function};
  \item $\mapsto \subset E\times E$ is a partial order called \emph{program order}, such that for all  $e\in E$, $\{e'\in E : e' \mapsto e\}$ is finite.
  \end{itemize}

\vspace{-1mm}
  Let $H = (U, Q, E, \Lambda, \mapsto)$ be a history. The sets $U_H = \{e\in E : \Lambda(e)\in U\}$ and $Q_H = \{e\in E : \Lambda(e)\in Q\}$ denote its sets of update and query events respectively. We also define some projections on the histories. The first one allows to withdraw some events: for $F\subset E$, $H_F = (U, Q, F, \Lambda|_F, \mapsto \cap (F\times F))$ is the history that contains only the events of $F$. The second one allows to substitute the order relation: if $\rightarrow$ is a partial order that respects the definition of a program order ($\mapsto$), $H^\rightarrow = (U, Q, E, \Lambda, \rightarrow \cap (E\times E))$ is the history in which the events are ordered by $\rightarrow$. Note that the projections commute, which allows the notation $H_F^\rightarrow$.
\end{definition}

\begin{definition}[Linearizations]
Let $H = (U, Q, E, \Lambda, \mapsto)$ be a distributed history. A linearization of $H$ corresponds to a sequential history that contains the same events as $H$ in an order consistent with the program order. More precisely, it is a word $\Lambda(e_0)\ldots\Lambda(e_n)\ldots$ such that $\{e_0, \ldots, e_n, \ldots\} =E$ and for all $i$ and $j$, if $i < j$, $e_j\not\mapsto e_i$. We denote by $\lin(H)$ the set of all linearizations of $H$.
\end{definition}

\begin{definition}[Consistency criterion]
A consistency criterion $C$ characterizes which histories are allowed for a given data type. It is a function $C$ that associates with any UQ-ADT $O$, a set of distributed histories $C(O)$. A shared object (instance of an UQ-ADT $O$) is $C$-consistent if all the histories it allows are in $C(O)$.
\end{definition}

\vspace{-3mm}
\section{Eventual Consistency}
\vspace{-3mm}
\label{section:eventual consistency}

In this section, we recall the definitions of eventual consistency \cite{vogels2008eventually} and strong eventual consistency \cite{shapiro2011comprehensive}. Fig. \ref{figure:histories} illustrates these two consistency criteria on small examples. In the remaining of this article, we consider an UQ-ADT $O = (U, Q_i, Q_o, S, s_0, T, G)$ and a history $H = (U, Q, E, \Lambda, \mapsto)$.

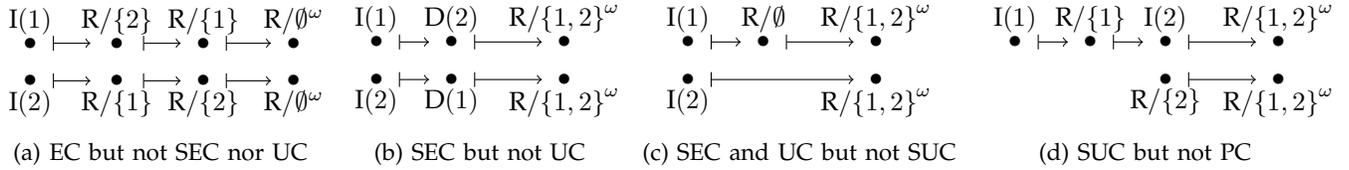
\begin{figure*}[t]
  \centering
  \hspace{\fill}
  \begin{subfigure}[b]{0.24\textwidth}
    \centering
    \begin{tikzpicture}
      \draw     (1,1)    node{$\bullet$} ;
      \draw     (1,1)    node[above]{$\text{I}(1)$} ;
      \draw[|->] (1.3,1) -- (1.8,1) ;
      \draw     (2.15,1)    node{$\bullet$} ;
      \draw     (2.15,1)    node[above]{$\text{R}/{\{2\}}$} ;
      \draw[|->] (2.5,1) -- (3,1) ;
      \draw     (3.3,1)    node{$\bullet$} ;
      \draw     (3.3,1)    node[above]{$\text{R}/{\{1\}}$} ;
      \draw[|->] (3.6,1) -- (4.2,1) ;
      \draw     (4.5,1)    node{$\bullet$} ;
      \draw     (4.5,1)    node[above]{$\text{R}/\emptyset^\omega$} ;
      
      \draw     (1,0.5)    node{$\bullet$} ;
      \draw     (1,0.5)    node[below]{$\text{I}(2)$} ;
      \draw[|->] (1.3,0.5) -- (1.8,0.5) ;
      \draw     (2.15,0.5)    node{$\bullet$} ;
      \draw     (2.15,0.5)    node[below]{$\text{R}/{\{1\}}$} ;
      \draw[|->] (2.5,0.5) -- (3,0.5) ;
      \draw     (3.3,0.5)    node{$\bullet$} ;
      \draw     (3.3,0.5)    node[below]{$\text{R}/{\{2\}}$} ;
      \draw[|->] (3.6,0.5) -- (4.2,0.5) ;
      \draw     (4.5,0.5)    node{$\bullet$} ;
      \draw     (4.5,0.5)    node[below]{$\text{R}/\emptyset^\omega$} ;
    \end{tikzpicture}
      \caption{EC but not SEC nor UC}
      \label{figure:ec rien d autre}
  \end{subfigure}
  \hspace{\fill}
  \begin{subfigure}[b]{0.2\textwidth}
    \centering
    \begin{tikzpicture}
      \draw     (1,1)    node{$\bullet$} ;
      \draw     (1,1)    node[above]{$\text{I}(1)$} ;
      \draw[|->] (1.3,1) -- (1.7,1) ;
      \draw     (2,1)  node{$\bullet$} ;
      \draw     (2,1)  node[above]{$\text{D}(2)$} ;
      \draw[|->] (2.3,1) -- (3.2,1) ;
      \draw     (3.5,1)  node{$\bullet$} ;
      \draw     (3.5,1)  node[above]{$\text{R}/{\{1, 2\}}^\omega$} ;
      
      \draw     (1,0.5)    node{$\bullet$} ;
      \draw     (1,0.5)    node[below]{$\text{I}(2)$} ;
      \draw[|->] (1.3,0.5) -- (1.7,0.5) ;
      \draw     (2,0.5)  node{$\bullet$} ;
      \draw     (2,0.5)  node[below]{$\text{D}(1)$} ;
      \draw[|->] (2.3,0.5) -- (3.2,0.5) ;
      \draw     (3.5,0.5)    node{$\bullet$} ;
      \draw     (3.5,0.5)    node[below]{$\text{R}/{\{1, 2\}}^\omega$} ;
    \end{tikzpicture}
      \caption{SEC but not UC}
      \label{figure:up_inc_history}
  \end{subfigure}
  \hspace{\fill}
  \begin{subfigure}[b]{0.24\textwidth}
    \centering
    \begin{tikzpicture}
      \draw     (1,1)    node{$\bullet$} ;
      \draw     (1,1)    node[above]{$\text{I}(1)$} ;
      \draw[|->] (1.3,1) -- (1.7,1) ;
      \draw     (2,1)  node{$\bullet$} ;
      \draw     (2,1)  node[above]{$\text{R}/{\emptyset}$} ;
      \draw[|->] (2.3,1) -- (3.2,1) ;
      \draw     (3.5,1)  node{$\bullet$} ;
      \draw     (3.5,1)  node[above]{$\text{R}/{\{1, 2\}}^\omega$} ;
      
      \draw     (1,0.5)    node{$\bullet$} ;
      \draw     (1,0.5)    node[below]{$\text{I}(2)$} ;
      \draw[|->] (1.3,0.5) -- (3.2,0.5) ;
      \draw     (3.5,0.5)    node{$\bullet$} ;
      \draw     (3.5,0.5)    node[below]{$\text{R}/{\{1, 2\}}^\omega$} ;
    \end{tikzpicture}
      \caption{SEC and UC but not SUC}
      \label{figure:tout sauf suc}
  \end{subfigure}
  \hspace{\fill}
  \begin{subfigure}[b]{0.24\textwidth}
    \centering
    \begin{tikzpicture}
      \draw     (1,1)    node{$\bullet$} ;
      \draw     (1,1)    node[above]{$\text{I}(1)$} ;
      \draw[|->] (1.3,1) -- (1.7,1) ;
      \draw     (2,1)  node{$\bullet$} ;
      \draw     (2,1)  node[above]{$\text{R}/{\{1\}}$} ;
      \draw[|->] (2.3,1) -- (2.7,1) ;
      \draw     (3,1)  node{$\bullet$} ;
      \draw     (3,1)  node[above]{$\text{I}(2)$} ;
      \draw[|->] (3.3,1) -- (4.2,1) ;
      \draw     (4.5,1)  node{$\bullet$} ;
      \draw     (4.5,1)  node[above]{$\text{R}/{\{1, 2\}}^\omega$} ;
      
      \draw     (3,0.5)    node{$\bullet$} ;
      \draw     (3,0.5)    node[below]{$\text{R}/{\{2\}}$} ;
      \draw[|->] (3.3,0.5) -- (4.2,0.5) ;
      \draw     (4.5,0.5)    node{$\bullet$} ;
      \draw     (4.5,0.5)    node[below]{$\text{R}/{\{1, 2\}}^\omega$} ;
    \end{tikzpicture}
      \caption{SUC but not PC}
      \label{figure:suc}
  \end{subfigure}
  \hspace{\fill}
  \caption{Four histories for an instance of $\mathcal{S}_{\mathbb{N}}$ (cf. example 1), with different
    consistency criteria. The arrows represent the program order, and an event labeled $\omega$
    is repeated an infinite number of times.}
  \label{figure:histories}
\end{figure*}

\vspace{-1mm}
\paragraph{Eventual consistency} eventual consistency requires that, if all the participants stop updating, all the replicas
eventually converge to the same state. In other word, $H$ is eventually consistent if it contains an infinite number of updates (i.e. the participants never stop writing) or if there exists a state (the consistent state) compatible with all but a finite number of queries.

\vspace{-1mm}
\begin{definition}[Eventual consistency]
A history $H$ is eventually consistent (EC) if $U_H$ is infinite or there exists a state $s\in S$ such that 
the set of queries that return non consistent values while in the state $s$, $\{q_i/q_o\in Q_H : G(s, q_i)\neq q_o\}$, 
is finite.
\end{definition}

\vspace{-2mm}
All the histories presented in Fig. \ref{figure:histories} are eventually consistent. The executions represent two processes sharing a set of integers. In Fig. \ref{figure:ec rien d autre}, the first process inserts value 1 and then reads twice the set and gets respectively $\{2\}$ and $\{1\}$; afterwards, it executes an infinity of read operations that return the empty set ($\omega$ in superscript denotes the operation is executed an infinity of times). In the meantime, the second process inserts a 2 then reads the set an infinity of times. It gets respectively $\{1\}$ and $\{2\}$ the two first times, and empty set an infinity of times. Both processes converge to the same state ($\emptyset$), so the history is eventually consistent. However, before converging, the processes can read anything a finite but unbounded number of times.

\vspace{-2mm}
\paragraph{Strong eventual consistency} strong eventual consistency requires that two replicas of the same object converge as soon as they have received the same updates. The problem with that definition is that the notions of replica and message reception are inherent to the implementation, and are hidden from the programmer that uses the object, so they should not be used in its specification. A visibility relation is introduced to model the notion of message delivery. This relation is not an order since it is not required to be transitive.

\vspace{-1mm}
\begin{definition}[Strong eventual consistency]
  A history $H$ is strong eventually consistent (SEC) if there exists an acyclic and reflexive relation
  $\xrightarrow{vis}$ (called \emph{visibility} relation) that contains $\mapsto$ and such that:
  \vspace{-2mm}
  \begin{itemize}
  \item Eventual delivery: when an update is viewed by a replica, it is eventually
    viewed by all replicas, so there can be at most a finite number of operations 
    that do not view it:\\
    $\forall u\in U_H, \{e\in E, u \hspace{1.5mm}\not\hspace{-1.5mm}\xrightarrow{vis} e \} ~\text{ is finite};$
  \item Growth: if an event has been viewed once by a process, it will remain visible forever:\\
    $\forall e, e', e''\in E, (e \xrightarrow{vis} e' \land e' \mapsto e'')
    \Rightarrow (e \xrightarrow{vis} e'');$
  \item Strong convergence: if two query operations view the same past of updates $V$, 
    they can be issued in the same state $s$:
    $\forall V \subset U_H, \exists s\in S, \forall q_i/q_o\in Q_H, $\\$V = \{u\in U_H : u\xrightarrow{vis} q_i/q_o\} \Rightarrow G(s,q_i) = q_o.$
  \end{itemize}
\end{definition}
\vspace{-2mm}
The history of Fig. \ref{figure:ec rien d autre} is not strong eventually consistent because the $\mathrm{I}(1)$ must be visible by all the queries of the first process (by reflexivity and growth), so there are only two possible sets of visible updates ($\{\mathrm{I}(1)\}$ and$\{\mathrm{I}(1), \mathrm{I}(2)\}$) for these events, but the queries are done in three different states ($\{1\}$, $\{2\}$ and $\emptyset$); consequently, at least two of these queries see the same set of updates and thus need to return the same value. Fig. \ref{figure:tout sauf suc}, on the contrary, is strong eventually consistent: the replicas that see $\{\mathrm{I}(1)\}$ are in state $\emptyset$ and those that see $\{\mathrm{I}(1), \mathrm{I}(2)\}$ are in state $\{1, 2\}$.

\vspace{-3mm}
\section{Pipelined Convergence}
\vspace{-3mm}
\label{section:pipelined convergence}

A straightforward way to strengthen eventual consistency is to compose it with another consistency criterion that imposes restrictions on the values that can be returned by a read operation. Causality is often cited as a possible candidate to play this role \cite{sun1998achieving}. As causal consistency is well formalized only for memory, we will instead consider Pipelined Random Access Memory (PRAM) \cite{lipton1988pram}, a weaker consistency criterion.
As the name suggests, PRAM was initially defined for memory. However, it can be easily extended to all UQ-ADTs. Let's call this new consistency criterion \emph{pipelined consistency (PC)}. In a pipelined consistent computation, each process must have a consistent view of its local history with all the updates of the computation. More formally, it corresponds to Def. \ref{def:PC}. Pipelined consistency is local to each process, as different processes can see concurrent updates in a different order.

\vspace{-1mm}
\begin{definition}\label{def:PC}
  A history $H$ is \emph{pipelined consistent} (PC) if, for all maximal chains 
  (i.e. sets of totally ordered events) $p$ of $H$, $\lin\left(H_{U_H \cup p}\right) \cap L(O)\neq \emptyset.$
\end{definition}
\vspace{-2mm}
Pipelined consistency can be implemented at a very low cost in wait-free systems. Indeed, it only requires FIFO reception. However, it does not imply convergence. For example, the history given in Figure \ref {figure:pc_pas_ec} is pipelined consistent but not eventually consistent. In this history, two processes $p_1$ and $p_2$ share a set of integers. Process $p_1$ first inserts $1$ and then $3$ in the set and then reads the set forever. Meanwhile, process $p_2$ inserts $2$, deletes $3$ and reads the set forever. The words $w_1$ and $w_2$ are correct linearizations for both processes, with regard to Definition \ref{def:PC} so the history is pipelined consistent, but after stabilization, $p_2$ sees the element $3$ whereas $p_1$ does not.

\begin{figure}[t]
  \begin{center}
  \begin{tikzpicture}
    \draw      (1,1)    node{$\bullet_a$} ;
    \draw      (1,1)    node[above]{$\text{I}(1)$} ;
    \draw[|->] (1.3,1) -- (2.2,1) ;
    \draw      (2.5,1)    node{$\bullet_b$} ;
    \draw      (2.5,1)    node[above]{$\text{I}(3)$} ;
    \draw[|->] (2.8,1) -- (3.7,1) ;
    \draw      (4,1)    node{$\bullet_c$} ;
    \draw      (4,1)    node[above]{$\text{R}/{\{1, 3\}}$} ;
    \draw[|->] (4.3,1) -- (5.7,1) ;
    \draw      (6,1)    node{$\bullet_d$} ;
    \draw      (6,1)    node[above]{$\text{R}/{\{1, 2, 3\}}$} ;
    \draw[|->] (6.3,1) -- (7.7,1) ;
    \draw      (8,1)    node{$\bullet_e$} ;
    \draw      (8,1)    node[above]{$\text{R}/{\{1, 2\}^\omega}$} ;

    \draw      (1,0.5)    node{$\bullet^f$} ;
    \draw      (1,0.5)    node[below]{$\text{I}(2)$} ;
    \draw[|->] (1.3,0.5) -- (2.2,0.5) ;
    \draw      (2.5,0.5)    node{$\bullet^g$} ;
    \draw      (2.5,0.5)    node[below]{$\text{D}(3)$} ;
    \draw[|->] (2.8,0.5) -- (3.7,0.5) ;
    \draw      (4,0.5)    node{$\bullet^h$} ;
    \draw      (4,0.5)    node[below]{$\text{R}/{\{2\}}$} ;
    \draw[|->] (4.3,0.5) -- (5.7,0.5) ;
    \draw      (6,0.5)    node{$\bullet^i$} ;
    \draw      (6,0.5)    node[below]{$\text{R}/{\{1, 2\}}$} ;
    \draw[|->] (6.3,0.5) -- (7.7,0.5) ;
    \draw      (8,0.5)    node{$\bullet^j$} ;
    \draw      (8,0.5)    node[below]{$\text{R}/{\{1, 2, 3\}^\omega}$} ;
  \end{tikzpicture}
  \end{center}
  $$ w_1 = \text{I}(1)\cdot \text{I}(3)\cdot \text{R}/{\{1, 3\}} \cdot
\text{I}(2)\cdot \text{R}/{\{1, 2, 3\}} \cdot \text{D}(3) \cdot
\text{R}/{\{1, 2\}^\omega} $$
  $$ w_2 = \text{I}(2)\cdot \text{D}(3)\cdot \text{R}/{\{2\}} \cdot
\text{I}(1)\cdot \text{R}/{\{1, 2\}} \cdot \text{I}(3) \cdot
\text{R}/{\{1, 2, 3\}^\omega} $$
  \caption{PC but not EC}
  \label{figure:pc_pas_ec}
\end{figure}
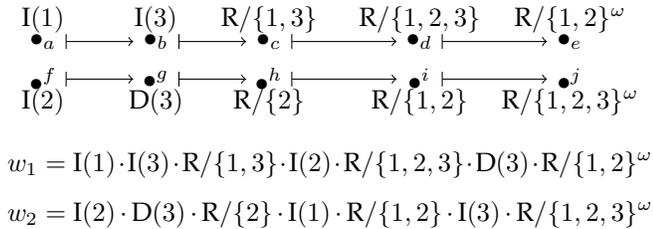

\vspace{-1mm}
\begin{proposition}[Implementation]
Pipelined convergence, that imposes both pipelined consistency and eventual consistency, cannot be implemented in a wait-free system. 
\end{proposition}
\vspace{-1mm}
\begin{proof}
We consider the same program as in Figure \ref{figure:pc_pas_ec}, and we suppose the shared set is pipelined convergent. By the same argument as developed in \cite{attiya1994sequential}, it is not possible to prevent the processes from not seeing each other's first update at their first reads. Indeed, if $p_1$ did not receive any message from process $p_2$, it is impossible for $p_1$ to make the difference between the case where $p_2$ crashed before sending any message and the case where all its messages were delayed. To achieve availability, $p_1$ must compute the return value based solely on its local knowledge, so it returns $\{1, 3\}$. Similarly, $p_2$ returns $\{2\}$. To circumvent this impossibility, it is necessary to make synchrony assumption on the system (e.g. bounds on transmission delays) or to assume the correctness of a majority of processes.

\vspace{-1mm}
If the first read of $p_1$ returns $\{1, 3\}$, as the set is pipelined consistent, there must exist a linearization for $p_1$ that contains all the updates, $\text{R}/{\{1, 3\}}$ and an infinity of queries. As $2 \not\in \{1, 3\}$, the possible linearizations are defined by the $\omega$-regular language $\text{I}(1)\cdot \text{I}(3)\cdot \text{R}/{\{1, 3\}}^+ \cdot \text{I}(2)\cdot \text{R}/{\{1, 2, 3\}}^\star \cdot \text{D}(3) \cdot \text{R}/{\{1, 2\}^\omega}$, so any history must contain an infinity of events labelled $\text{R}/{\{1, 2\}^\omega}$. Similarly, if $p_2$ starts by reading $\{2\}$, it will eventually read $\{1, 2, 3\}$ an infinity of times. This implies that pipelined convergence cannot be provided in wait-free systems.
\end{proof}

\vspace{-1mm}
Consequently causal consistency, that is stronger than pipelined consistency, cannot be satisfied together with eventual consistency in a wait-free system.

\section{Update Consistency}
\label{section:update consistency}

In this section, we introduce two new consistency criteria: update consistency and strong update consistency\footnote{These consistency criteria were previously presented as a brief announcement in DISC 2014 \cite{perrinbrief}.}, and we compare them to eventual consistency and strong eventual consistency. Fig. \ref{figure:histories} illustrates these four consistency criteria on four small examples.

\paragraph{Update consistency} eventual consistency and strong eventual consistency are not interested in defining the states that are reached during the histories (the same updates have to lead to the same state whatever is the state). They do not depend on the sequential specification of the object, so they give very little constraints on the histories. For example, an implementation that ignores all the updates is strong eventually consistent, as all the queries return the initial state. In update consistency, we impose the existence of a total order on the updates, that contains the program order and that leads to the consistent state according to the abstract data type. Another equivalent way to approach update consistency is that, if the number of updates is finite, it is possible to remove a finite number of queries such that the history is sequentially consistent.

\begin{definition}[Update consistency]
A history $H$ is update consistent (UC) if $U_H$ is infinite or if there exists a 
finite set of queries $Q' \subset Q_H$ such that 
$\lin\left(H_{E\setminus Q'}\right) \cap L(O) \neq \emptyset$.
\end{definition}

The history of Fig. \ref{figure:tout sauf suc} is update consistent because the sequence of operations $\mathrm{I}(1)\mathrm{I}(2)$ 
is a possible explanation for the state $\{1, 2\}$. The history of Fig. \ref{figure:up_inc_history} is not update consistent because 
any linearization of the updates would position a deletion as the last event. Only three consistent states are actually possible:
state $\emptyset$, e.g. for the linearization $\mathrm{I}(1)\cdot\mathrm{I}(2)\cdot\mathrm{D}(1)\cdot\mathrm{D}(2)$, 
state $\{1\}$ for the linearization $\mathrm{I}(2)\cdot\mathrm{D}(1)\cdot\mathrm{I}(1)\cdot\mathrm{D}(2)$ and state $\{2\}$
for the linearization $\mathrm{I}(1)\cdot\mathrm{D}(2)\cdot\mathrm{I}(2)\cdot\mathrm{D}(1)$. 
Update consistency is incomparable with strong eventual consistency.

\pagebreak
\paragraph{Strong update consistency} strong update consistency is a strengthening of both update consistency and strong
eventual consistency. The relationship between update consistency and strong update
consistency is analogous to the relation between eventual consistency and strong 
eventual consistency.

\vspace{-1mm}
\begin{definition}[Strong update consistency]
A history $H$ is strong update consistent (SUC) if there exists (1) an acyclic and reflexive 
relation $\xrightarrow{vis}$ that contains $\mapsto$ and (2) a total order $\le$ that contains $\xrightarrow{vis}$ such that:
\vspace{-1mm}
  \begin{itemize}
  \item Eventual delivery: \\$\forall u\in U_H, \{e\in E, u \hspace{1.5mm}\not\hspace{-1.5mm}\xrightarrow{vis} e \} \text{ is finite};$
  \item Growth: \\$\forall e, e', e''\in E, \left(e \xrightarrow{vis} e' \land e' \mapsto e''\right) \Rightarrow \left(e \xrightarrow{vis} e''\right);$
  \item Strong sequential convergence: A query views an update if this update precedes it according to
    $\xrightarrow{vis}$. Each query is the result of the ordered execution, according to $\le$, of the updates it views:\\
    $\forall q\in Q_H, \lin\left(H_{V(q) \cup \{q\}}^\le\right) \cap L(O) \neq \emptyset$ \\where $V(q) = \{u\in U_H : u \xrightarrow{vis} q\}.$
  \end{itemize}
\end{definition}

\vspace{-1mm}
Fig. \ref{figure:suc} shows an example of strong update consistent history: 
nothing prevents the second process from seeing the insertion of $2$ before that of $1$. 
Strong eventual consistency and update consistency does not imply strong update consistency: 
in the history of Fig. \ref{figure:tout sauf suc}, after executing event $\mathrm{I}(1)$, 
the only three possible update linearizations are $\mathrm{I}(1)$, $\mathrm{I}(1)\cdot\mathrm{I}(2)$ 
and $\mathrm{I}(2)\cdot\mathrm{I}(1)$ and none of them can lead to the state $\emptyset$ according to 
the sequential specification of a set object. So the history of Fig. \ref{figure:tout sauf suc} 
is not strong update consistent, while it is update consistent and strong eventually consistent.

\vspace{-1mm}
\begin{proposition}[Comparison of consistency criteria]
If a history $H$ is update consistent, then it is eventually consistent. If $H$ is strong update consistent, then it is both   strong eventually consistent and update consistent.
\end{proposition}

\vspace{-2mm}
\begin{proof}
Suppose $H$ is update consistent. If $H$ contains an infinite number of updates, then it
is eventually consistent. Otherwise, there exists a finite set $Q'\subset Q_H$ and a 
word $w\in\lin\left(H_{E_H\setminus Q'}\right) \cap L(O)$. As the number of updates is finite, 
there is a finite prefix $v$ of $w$ that contains them all. $v\in L(O)$, so it labels a 
path between $s_0$ and a state $s$ in the UQ-ADT. All the queries that are in $w$ but not 
in $v$ return the same state $s$, and the number of queries in $Q'$ and $v$ is finite. 
Hence, $H$ is eventually consistent.

\vspace{-1mm}
Suppose $H$ is strong update consistent with a finite number of updates.
\linebreak[4] $Q' = \bigcup_{u\in U_H}\{q\in Q_H, q \le u\}$ is finite, and $\lin\left(E_H\setminus Q'\right)$ 
contains only one word that is also contained into $L(O)$. Obviously, $H$ is update consistent

\vspace{-1mm}
Now, suppose $H$ is strong update consistent. Strong update consistency 
respects both eventual delivery and growth properties. Let $V \subset U_H$. As the relation $\le$ is a total order, 
there is a unique word $w$ in $\lin\left(H_V^\le\right) \cap L(O)$. Let us denote $s$ the state obtained after the execution of $w$. 
For all $q\in Q_H$ such that $V = \{u\in U_H : u\xrightarrow{vis} q\}$, $\lin\left(H_{E_q \cup \{q\}}^\le\right) \cap L(O) = \{w\cdot\Lambda(q)\}$, 
so $q=q_i/q_o$ with $G(s, q_i) = q_o$. Consequently, $H$ is strong eventually consistent.
\end{proof}

\section{Expressiveness of Update Consistency: a Case Study}
\label{section:set}

The set is one of the most studied eventually consistent data structures. Different types of sets have been proposed as extensions to CRDTs to implement eventually consistent sets even though the insert and delete operations do not commute. The simplest set is the Grow-Only Set (G-Set) \cite {shapiro2011comprehensive}, in which it is only possible to insert elements. As the insertion of two elements commute, G-Set is a CRDT. Using two G-Set, a white list for inserted elements and a black list for the deleted ones, it is possible to build a Two-Phases Set (2P-Set, a.k.a. U-Set, for Unique Set) \cite{wuu1986efficient}, in which it is possible to insert and remove elements, but never insert again an element that has already been deleted. Other implementations such as C-Set \cite {aslan2011c} and PN-Set, add counters on the elements to determine if they should be present or not. The Observe-Remove Set (OR-Set) \cite{shapiro2011comprehensive,mukund2014optimized} is the best documented algorithm for the set. It is very close to the 2P-Set in its principles, but each insertion is timestamped with a unique identifier, and the deletion only black-lists the identifiers that it observes. It guaranties that, if an insertion and a deletion of the same element are concurrent, the insertion will win and the element will be added to the set. Finally, the last-writer-wins element set (LWW-element-Set) \cite {shapiro2011comprehensive} attaches a timestamp to each element to decide which operation should win in case of conflict. All these sets, and the eventually consistent objects in general, have a different behavior when they are used in distributed programs.

The above mentioned implementations are eventually consistent. However, as eventual consistency does not impose a semantic 
link between updates and queries, it is hazardous to say anything on the conformance to the specification of the object. 
Burckhardt \emph{et al.} \cite{burckhardt2014replicated} propose to specify the semantics of a query by a function 
on its concurrent history, called \emph{visibility}, that corresponds to the visibility relation in strong eventual consistency,
and a linearization of this history, called \emph{arbitration}. 
In comparison, sequential specifications are restricted to the arbitration relation. It implies that fewer update consistent objects 
than eventually consistent objects can be specified. Although the variety of objects with a distributed specification seems to be a 
chance that compensates the lower level of abstraction it allows, an important bias must be taken into account: from the point of 
view of the user, the visibility of an operation is not an \emph{a priori} property of the system, but an \emph{a posteriori} 
way to explain what happened. If one only focuses on the final state, an update consistent object is appropriate to be used 
instead of an eventually consistent object, since the final state is the same as if no operations were concurrent.

By adding further constraints on the histories, concurrent specifications strengthen the consistency criteria. Even if strong update consistency is stronger than strong eventual consistency, we cannot say in general that a strong update consistent object can always be used instead of its strong eventually 
consistent counterpart. We claim that this is true in practice for \emph{reasonable} objects, and we prove this in the case of the Insert-wins set (the concurrent specification of the OR-set). The arbitration relation is not used for the OR-set, and the visibility relation has already been defined for strong eventual consistency. The concurrent specification only adds one more constraint on this relation: an element is present in the set 
if and only if it was inserted and is not yet deleted.

\begin{definition}[Strong eventual consistency for the Insert-wins set]
A history $H$ is strong eventually consistent for the Insert-wins set on a support $\mathit{Val}$ if it is strong eventually consistent for the set $\mathcal{S}_{\mathit{Val}}$ and the visibility relation $\xrightarrow{vis}$ verifies the following additional property. For all $x\in \mathit{Val}$ and $q\in Q_H$, with $\Lambda(q) = R/s$, $x\in s \Leftrightarrow \left(\exists u\in vis(q, \mathrm{I}(x)), \forall u'\in vis(q, \mathrm{D}(x)), u\hspace{1.5mm}\not\hspace{-1.5mm}\xrightarrow{vis} u'\right)$, where for all  $o \in U$, $vis(q, o) = \{u\in U_H : u \xrightarrow{vis} q \land \Lambda(u) = o\}$.
\end{definition}

The OR-Set implementation of a set is not update consistent. The history on Fig. \ref{figure:up_inc_history} is not update consistent, as the last operation must be a deletion. However, if the updates made by a process are not viewed by the other process before it makes its own updates, the insertions will win and the OR-set will converge to $\{1, 2\}$. On the contrary, a strong update consistent implementation of a set can always be used instead of an Insert-wins set, as it only forbids more histories.

\begin{proposition}[Comparison with Insert-wins set]
  Let $H = (U, Q, E, \Lambda, \mapsto)$ be a history that is strong update
consistent for $\mathcal{S}_{\mathit{Val}}$. 
  Then $H$ is strong eventually consistent for the Insert-wins set.
\end{proposition}

\pagebreak

\begin{proof}
  Suppose $H$ is strong update consistent for $\mathcal{S}_{\mathit{Val}}$. 
  We define the new relation $\xrightarrow{IW}$ such that 
  for all $e, e' \in E$, $e\xrightarrow{IW}e'$ if one of the following conditions
  holds:
  \begin{itemize}
    \item $e\xrightarrow{vis} e'$;
    \item $e$ and $e'$ are two updates on the same element and $e\le e'$;
    \item $e'$ is a query, and there is an update $e''$ such that $e\xrightarrow{IW}
      e''$ and $e''\xrightarrow{IW} e'$.
  \end{itemize}
The relation $\xrightarrow{IW}$ is acyclic because it is included in $\le$, its growth and eventual delivery properties are ensured by the fact that it contains $\xrightarrow{vis}$. Moreover, no two updates for the same element are concurrent according to $\xrightarrow{IW}$ and the  last updates are also the last for the $\le$ relation, consequently $H$ is strong eventually consistent for the Insert-wins set.
\end{proof}

This result implies that an OR-set can always be replaced by an update consistent set, because the guaranties it ensures are weaker than those of the update consistent set. It does not mean that the OR-set is worthless. It can be seen as a cache consistent set \cite{goodman1991cache} that, in some cases may have a better space complexity than update consistency.

\vspace{-2mm}
\section{Generic Construction of Strong Update Consistent Objects}
\vspace{-2mm}
\label{section:implementation}

In this section, we give a generic construction of strong update consistent objects in crash-prone asynchronous message-passing systems. This construction is not the most efficient ever as it is intended to work for any UQ-ADT object in order to prove the universality of update consistency. For a specific object an ad hoc implementation on a specific system may be more suitable.

\vspace{-2mm}
\subsection{System Model}
\vspace{-2mm}

We consider a message-passing system composed of finite set of sequential
\emph{processes} that may fail by halting. A faulty process simply stops operating.
A process that does not crash during an execution is correct. We make no assumption on
the number of failures that can occur during an execution. Processes communicate by
exchanging \emph{messages} using a communication network complete and reliable. A
message sent by a correct process to another correct process is eventually received.
The system is asynchronous; there is no bound on the relative speed of processes nor
on the message transfer delays. In such a situation a process cannot wait for the
participation of any a priori known number of processes as they can fail.
Consequently, when an operation on a replicated object is invoked locally at some
process, it needs to be completed based solely on the local knowledge of the process. 
We call this kind of systems wait-free asynchronous message-passing system.

\vspace{-1mm}
We model executions as histories made up of the sequences of events generated by the different processes. As we focus on shared objects and their implementation, only two kinds of actions are considered: the operations on shared objects, that are seen as events in the distributed history, and message receptions.

\vspace{-2mm}
\subsection{A universal implementation}
\vspace{-2mm}

\begin{algorithm}[t]
  \SetKw{Var}{var}
  \SetKw{Fun}{fun}
  \SetKw{Receive}{on receive}
  \SetKw{Broadcast}{broadcast}

  \SetKw{True}{true}
  \SetKw{False}{false}

  \SetKwData{Kwset}{update}
  \SetKwData{Kwstate}{state}
  \SetKwData{Kwclock}{clock}

  \SetKwData{Kwu}{u}
  \SetKwData{Kwq}{q}
  \SetKwData{Kwc}{cl}
  \SetKwData{Kwj}{j}
  \SetKwData{Kws}{state}

  \SetKwFunction{KwPid}{pid}
  
  \SubAlgo{\textbf{object} $(U, Q_i, Q_o, S, s_0, T, G)$}{
    \Var $\Kwclock_i \in \mathbb{N} \leftarrow 0$\;
    \Var $\Kwset_i \subset \left(\mathbb{N} \times \mathbb{N} \times U\right) \leftarrow \emptyset$\;

    \SetKwFunction{Kwapply}{update}
    \SetKwFunction{KwMessage}{message}
    \SubAlgo{\Fun \Kwapply $(\Kwu \in U)$}{
      $\Kwclock_i \leftarrow \Kwclock_i + 1$\;
      \Broadcast \KwMessage $(\Kwclock_i, i, \Kwu)$\;
    }

    \SetKwFunction{Kwinsert}{insertSort}
    \SubAlgo{\Receive \KwMessage $(\Kwc \in \mathbb{N}, \Kwj \in \mathbb{N}, \Kwu \in Q)$}{
      $\Kwclock_i \leftarrow \max(\Kwclock_i, \Kwc)$\;
      $\Kwset_i \leftarrow \Kwset_i \cup \{(\Kwc, \Kwj, \Kwu)\}$\;
    }

    \SetKwFunction{Kwread}{query}
    \SubAlgo{\Fun \Kwread $(\Kwq\in Q_i) \in Q_o$}{
      $\Kwclock_i \leftarrow \Kwclock_i + 1$\;
      \Var $\Kws_i \in S \leftarrow s_0$\;
      \For{$(\Kwc, \Kwj, \Kwu) \in \Kwset_i$ sorted on $(\Kwc, \Kwj)$}{
        $\Kws_i \leftarrow T(\Kws_i, \Kwu)$\;
      }
      \Return $G(\Kws_i, \Kwq)$\;
    }
  }
  \caption{a generic UQ-ADT (code for $p_i$)}
  \label{algo:transition system}
\end{algorithm}

\vspace{-1mm}
Now, we prove that strong update consistency is universal, in the sense that every UQ-ADT
has a strong update consistent implementation in a wait-free asynchronous system. 
Algorithm \ref{algo:transition system} presents an implementation of a generic UQ-ADT. The principle is to 
build a total order on the updates on which all the participants agree, and then to rewrite the history \emph{a
posteriori} so that every replica of the object eventually reaches the state corresponding to the
common sequential history. Any strategy to build the total order on the updates
would work. In Algorithm \ref{algo:transition system}, this order is built from a Lamport's
clock \cite{lamport1978time} that contains the happened-before precedence relation. Process order is hence respected. 
A logical Lamport's clock is a pre-total order as some events may be associated with the same logical time. 
In order to have a total order, the events are timestamped with a pair composed of the logical time and the 
id of the process that produced it (process ids are assumed unique and totally ordered). The algorithm
actions performed by a process $p_i$ are atomic and totally ordered by an order $\mapsto_i$.
The union of these orders for all processes is the program order $\mapsto$. 

\vspace{-1mm}
At the application level, a history is composed of update and query operations. In
order to allow only strong update consistent histories, Algorithm \ref{algo:transition system} proposes
a procedure $update()$ and a function $query()$. A history $H$ is allowed
by the algorithm if $update(u)$ is called each time a process performs an update $u$,
and $query(q_i)$ is called and returns $q_o$ when the event $q_i/q_o$ appears in the history.
The code of Algorithm \ref{algo:transition system} is given for process $p_i$.
Each process $p_i$ manages its view $clock_i$ of the logical clock and a list
$updates_i$ of
all timestamped update events process $p_i$ is aware of. The list $updates_i$ contains
triplets $(cl,j,u)$ where $u$ is an update event and $(cl,j)$ the associated
timestamp. This list is sorted according to the timestamps of the updates:
$(cl,j)<(cl',j')$ if ($cl<cl'$) or ($cl=cl'$ and $j<j'$).

\vspace{-1mm}
The algorithm timestamps all events (updates and queries). When 
an update is issued locally, process $p_i$ informs all the other 
processes by reliably broadcasting a message to all other processes 
(including itself). Hence, all processes will eventually be aware of 
all updates. When a $message(cl, j, u)$ is received, $p_i$ updates its clock and 
inserts the event to the list $updates_i$.
When a query is issued, the function $query()$ replays locally the whole list of update 
events $p_i$ is aware of starting from the initial state then it executes the query on the
state it obtains.

\vspace{-1mm}
Whenever an operation is issued, its is completed without waiting for any other process.
This corresponds to wait-free executions in shared memory distributed systems and
implies fault-tolerance.

\vspace{-1mm}
\begin{proposition}[Strong update consistency]
  All histories allowed by Algorithm \ref{algo:transition system} are strong update consistent. 
\end{proposition}
\begin{proof}
  Let $H = (U, Q, E, \Lambda, \mapsto)$ be a distributed history allowed by Algorithm \ref{algo:transition system}. 
  Let $e, e' \in E_H$ be two operations invoked by processes $p_i$ and $p_{i'}$, 
  on the states $(\verb+update+, \verb+clock+)$ and $(\verb+update+', \verb+clock+')$, respectively. We pose:
  \vspace{-1mm}
  \begin{itemize}
  \item $e \xrightarrow{vis} e'$ if $e\in U_H$ and $p_{i'}$ received the message sent during the execution of $e$ 
    before it starts executing $e'$, or $e\in Q_H$ and $e\mapsto e'$. As the messages are received instantaneously 
    by the sender, $\xrightarrow{vis}$ contains $\mapsto$. It is growing because the set of messages received by a 
    process is growing with time.
\pagebreak
  \item $e\le e'$ if $c < c'$ or $c=c'$ and $i < i'$. This lexical order is total
    because two operations on the same 
    process have a different clock. Moreover it contains $\xrightarrow{vis}$ because when
    $p_{i'}$ received the message sent by $e$,
    it executed line 9 and when it executed $e'$, it executed line 5, so $c' \ge
    c+1$. 
    Moreover, the history of $e$ contains at most $c\times n + i$ events, where $n$
    is the number of processes, so it is finite.
  \end{itemize}
  \vspace{-2mm}
  Let $q\in Q_H$ and $E_q = \{u\in U_H : u \xrightarrow{vis} q\}$. 
  Lines 15 to 18 build an explicit sequential execution, that is in $\lin\left(H_{E_q \cup \{q\}}^\le\right)$ 
  by definition of $\le$ and in $L(O)$ by definition of $O$.
\end{proof}

\vspace{-2mm}
\subsection{Complexity}
\vspace{-2mm}

Algorithm \ref{algo:transition system} is very efficient in terms of 
network communication. A unique message is broadcast for each update and 
each message only contains the information
to identify the update and a timestamp composed of two integer values, that
only grow logarithmically with the number of processes and the number of 
operations. Moreover, this algorithm is wait-free and its execution does not depend on
the latency of the network.

  \vspace{-1mm}
This algorithm re-executes all past updates each time a new query is issued. In an effective implementation, a process can keep intermediate states. These intermediate states are re-computed only if very late message arrive. The algorithm does not look space efficient also as the whole history must be kept in order to rebuild a sequential history. 
Because data space is cheap and fast nowadays, compared to bandwidth, many applications can afford this complexity and would keep this information anyway. For example, banks keep track of all the operations made on an account for years for legal reasons. In databases systems, it is usual to record all the events in log files. Moreover, asynchrony is used as a convenient abstraction for systems in which transmission delays are actually bounded, but the bound is too large to be used in practice. This means that after some time old messages can be garbage collected.

\begin{algorithm}[t]
  \SetKw{Var}{var}
  \SetKw{Fun}{fun}
  \SetKw{Receive}{on receive}
  \SetKw{Broadcast}{broadcast}

  \SetKwData{Kwstate}{mem}
  \SetKwData{Kwclock}{clock}

  \SetKwData{Kwc}{cl}
  \SetKwData{Kwj}{j}
  \SetKwData{Kwx}{x}
  \SetKwData{Kwv}{v}
  \SetKwData{Kws}{state}

  \SetKwFunction{KwPid}{pid}
  
  \SubAlgo{\textbf{object} $\text{UC\_mem}(X, V, v_0)$}{
    \Var $\Kwclock_i \in \mathbb{N} \leftarrow 0$\;
    \Var $\Kwstate_i \in \text{mem}(X, (\mathbb{N}^2 \times V), (0, 0, v_0))$\;

    \SetKwFunction{Kwapply}{write}
    \SetKwFunction{KwMessage}{msg}
    \SetKwFunction{Kwread}{read}
    \SubAlgo{\Fun \Kwapply $(\Kwx \in X, \Kwv \in V)$}{
      $\Kwclock_i \leftarrow \Kwclock_i + 1$\;
      \Broadcast \KwMessage $(\Kwclock_i, i, \Kwx, \Kwv)$\;
    }
    \SubAlgo{\Receive \KwMessage $(\Kwc \in \mathbb{N}, \Kwj \in \mathbb{N}, \Kwx \in X, \Kwv \in V)$}{
      $\Kwclock_i \leftarrow \max(\Kwclock_i, \Kwc)$\;
      \Var $(\Kwc', \Kwj', \Kwv')\in \mathbb{N}^2\times V \leftarrow \Kwstate_i.\Kwread(\Kwx)$\;
      \If{$(\Kwc', \Kwj') < (\Kwc, \Kwj)$} {
        $\Kwstate_i.\Kwapply(\Kwx, (\Kwc, \Kwj, \Kwv))$
      }
    }

    \SubAlgo{\Fun \Kwread $(\Kwx \in X) \in V$}{
      \Var $(\Kwc, \Kwj, \Kwv)\in \mathbb{N}^2\times V \leftarrow \Kwstate_i.\Kwread(\Kwx)$\;
      \Return $\Kwv$\;
    }
  }
  \caption{the shared memory (code for $p_i$)}
  \label{algo:memory}
\end{algorithm}

  \vspace{-1mm}
The proposed algorithm is a theoretical work whose goal is to prove that any update-query object has a strong update consistent implementation. This genericity prevents an effective implementation that may take benefit from the nature and the specificity of the actual object. The best example of this are pure CRDTs like the counter and the grow-only set. If all the update operations commute in the sequential specification, all linearizations would lead to the same state so a naive implementation, that applies the updates on a replica as soon as the notification is received, achieves update consistency. In \cite{KBL93}, Karsenty and Beaudouin-Lafon propose an algorithm to implement objects such that each update operation $u$ contains an \emph{undo} $u^{-1}$ such that for all $s$, $T(T(s, u), u^{-1}) = s$. This algorithm is very close to ours as it builds the convergent state from a linearization of the updates stored by each replica. They use the undo operations to position newly known updates at their correct place, which saves computation time. As it is a very frequent example in distributed systems, we now focus on the shared memory object.

  \vspace{-1mm}
Algorithm \ref{algo:memory} shows an update consistent implementation of the shared memory object. 
A shared memory offers a set $X$ of registers that contain values taken from a set $V$. 
The query operation $\text{read}(x)$, where $x\in X$, returns the last value $v\in V$ 
written by the update operation $\text{write}(x, v)$, or the initial value $v_0\in V$ if $x$ was never written. 
Algorithm \ref{algo:memory} orders the updates exactly like Algorithm \ref{algo:transition system}. 
As the old values can never be read again, it is not necessary to store them forever, so the algorithm 
only keeps in memory the last known value of each register and its timestamp in a local memory $mem_i$, implemented with an associative array. 
When a process receives a notification for a write, it updates its local state if the value is newer that the current one, 
and the read operations just return the current value. This implementation only needs constant computation time for both 
the reads and the writes, and the complexity in memory only grows logarithmically with time and the number of participants.

\vspace{-2mm}
\section{Conclusion}\label{section:conclusion}
\vspace{-3mm}

This paper proposes a new consistency criterion, update consistency, that is stronger than eventual consistency and weaker than sequential consistency. Our approach formalizes the intuitive notions
of sequential specification for an abstract data type and distributed history. This formalization first allowed to prove that eventual consistency when associated with causal consistency or PRAM consistency can no more be implemented in an asynchronous distributed system where all but one process may crash. 

\vspace{-1mm}
This paper formalizes the new consistency criterion and proves that (1) it is strictly stronger than eventual consistency and (2) that it is universal in the sense that allowed any update consistent object can be implemented in wait-free systems. The latter has been proved through a generic construction that implement all considered data types.

\vspace{-2mm}
\section*{Acknowledgment}\label{section:ack}
\addcontentsline{toc}{section}{Acknowledgment} 
\vspace{-3mm}

This work has been partially supported by a French government support granted to the CominLabs excellence laboratory (Project \emph{DeSceNt: Plug-based Decentralized Social Network}) and managed by the French National Agency for Research (ANR) in the "Investing for the Future" program under reference Nb. ANR-10-LABX-07-01.
This work was partially funded by the French ANR project SocioPlug (ANR-13-INFR-0003).

\vspace{-2mm}
\bibliographystyle{IEEEtran}
\bibliography{ipdpsUC}

\end{document}